\documentclass[draftcls, onecolumn, 12pt]{IEEEtran}
\usepackage{amsmath,amssymb}
\usepackage[dvips]{graphicx}
\usepackage{amsfonts}
\usepackage[mathscr]{eucal}
\usepackage{latexsym}
\usepackage{amsthm}
\usepackage{exscale}
\usepackage[mathscr]{eucal}
\usepackage{bm}
\usepackage[dvipsnames]{color}
\usepackage{epsfig}
\usepackage{algorithm}
\usepackage{algorithmic}
\usepackage[center,small]{caption}

\newtheorem{theorem}{Theorem}
\newtheorem{remark}{Remark}
\newtheorem{lemma}{Lemma}
\IEEEoverridecommandlockouts
\begin{document}
\title{Improved Linear Precoding over Block Diagonalization in Multi-cell Cooperative Networks}
\author{Yong~Zeng, Erry~Gunawan and Yong~Liang~Guan\\
\thanks{The authors are with the School of Electrical and Electronic Engineering, Nanyang Technological University, Singapore 639801
(email: \{ze0003ng, egunawan, eylguan\}@e.ntu.edu.sg)}
\thanks{This work was supported in part by the Advanced Communications Research Program DSOCL06271, a research grant from the Directorate of Research and Technology (DRTech), Ministry of Defence, Singapore.}}

\maketitle
\begin{abstract}
In downlink multiuser multiple-input multiple-output (MIMO) systems, block diagonalization (BD) is a practical linear precoding scheme which achieves the same degrees of freedom (DoF) as the optimal linear/nonlinear precoding schemes. However, its sum-rate performance is rather poor in the practical SNR regime due to the transmit power boost problem. In this paper, we propose an improved linear precoding scheme over BD with a so-called ``effective-SNR-enhancement'' technique. The transmit covariance matrices are obtained by firstly solving a power minimization problem subject to the minimum rate constraint achieved by BD, and then properly scaling the solution to satisfy the power constraints. It is proved that such approach equivalently enhances the system SNR, and hence compensates the transmit power boost problem associated with BD. The power minimization problem is in general non-convex. We therefore propose an efficient algorithm that solves the problem heuristically. Simulation results show significant sum rate gains over the optimal BD and the existing minimum mean square error (MMSE) based precoding schemes.          \end{abstract}
\begin{IEEEkeywords}Linear precoding,   block diagonalization, network MIMO, multi-cell cooperation, per-base-station power constraint, convex optimization
\end{IEEEkeywords}
\section{Introduction}
Traditional approaches for downlink inter-cell interference management, such as frequency reuse, coordinated scheduling or beamforming techniques \cite{210}, mostly follow the notion of ``interference avoidance''. Recent work on multi-cell cooperative processing (MCP) \cite{211}, with the idea of exploiting the interfering links instead of simply avoiding them, shows that the spectral efficiency can be significantly enhanced by allowing joint transmission from the interfering base stations (BS). In principle, MCP transforms the multi-cell multi-user network into a giant multi-user system, where the resources can be more efficiently utilized. In the ideal case, downlink MCP enabled networks are equivalent to broadcast channels (BC), where dirty-paper coding (DPC) is capacity achieving~\cite{212}. However, DPC is generally too complex for practical implementation for real-time systems due to its complicated nonlinear encoding and decoding processes. As a consequence, linear precoding schemes have drawn a lot of attentions since they can achieve a reasonable balance between complexity and performance \cite{214,215,213,201}. One class of linear precoding schemes of particular interest is block diagonalization (BD), which can be viewed as an extension of zero-forcing channel inversion in the multiple-input single-output (MISO) broadcast channels, e.g., \cite{174, 221}, to the more general multiuser MIMO networks. With BD, the inter-user interference is completely eliminated by restricting the precoding matrix for each mobile station (MS) to be orthogonal to the channels associated with all other MSs. The initial study on BD mostly focuses on single-cell systems, where the sum-power constraint is generally considered \cite{216,182,217,218}. The extension to multi-cell networks with per-BS power constraints is non-trivial \cite{219,220}. In \cite{220}, the weighted sum rate maximization problem with BD was formulated as a convex optimization problem, from which a closed form expression for the optimal BD precoders was derived. The main advantages of BD lie on its simplicity and good performance at high SNR. However, it gives quite poor performance in the low to medium SNR regime due to the transmit power boost problem.

 One straightforward solution to improving the low-to-medium SNR performance of BD seems to be the MMSE-based precoding schemes. For the special case of single-antenna receivers, a regularized channel inversion scheme was proposed in \cite{174}, with the regularization parameter inversely proportional to SNR. Such techniques were extended to the multiuser MIMO systems with sum-power constraint \cite{222,223}. For multi-cell cooperative networks with per-BS power constraints, the authors in \cite{224} proposed to decompose the precoding matrix into a preliminary matrix and a diagonal power control matrix, where the preliminary matrix was designed to have the MMSE structure in order to balance the noise and interference effects. Another MMSE-based precoding scheme under per-BS power constraints was proposed in \cite{225}, where sum-MSE is minimized directly. However, due to the complicated mathematical structure, only a local optimal solution can be obtained and it requires iteratively solving a sequence of convex problems. As will be shown in Section~\ref{S:simulation}, under per-BS power constraints, although the MMSE-based precoding schemes can provide certain performance gain over BD at low SNR, the achievable sum rates are lower than that achieved by BD as SNR increases. In other words, the existing MMSE-based precoding algorithms fail to achieve the same DoF as BD.

In this work, we focus on the MCP-enabled downlink networks under per-BS power constraints. The main objective is to propose an efficient scheme that improves the performance of BD in the low to medium SNR regime, while preserving its good performance at high SNR. Unlike BD, the proposed scheme takes the noise effect into consideration and interference leakage is allowed. The performance gain is mainly attributed to a so-called effective-SNR-enhancement technique, by solving a power minimization problem with a minimum rate constraint achieved by BD and properly scaling the obtained transmit covariance matrices to satisfy the power constraint. Such technique provides a method to compensate the transmit power boost problem associated with BD. The power minimization problem is non-convex in general due to the non-convex rate and rank constraints. To tackle this issue, we firstly convexify the rate constraints with Taylor approximation and then solve the \emph{rank-relaxed} convexified problem in the dual domain. A closed form solution in terms of the dual variables is then obtained. With such an expression, it is found that the solution is also optimal to the \emph{rank-constrained} non-convex problem since it automatically satisfies the rank constraints. The proposed scheme is efficient since eventually only one convex optimization problem needs to be solved. 

The rest of the paper is organized as follows. Section~\ref{S:systemModel} introduces the system model and problem formulation. Section~\ref{S:BD} reviews the optimal BD under per-BS power constraints. Section~\ref{S:proposed} presents the proposed scheme and in Section~\ref{S:simulation},  numerical results are given. Finally, conclusions are given in Section~\ref{S:conclusions}.

 \emph{Notations:} Throughout this paper, scalars are denoted by italicized letters. Boldface lower- and upper-case letters denote vectors and matrices, respectively. $\mathbf{I}$ denotes the identity matrix and $\mathbf{0}$ denotes an all-zero matrix. For a square matrix $\mathbf{S}$, $\mathrm{Tr}(\mathbf{S})$, $|\mathbf{S}|$, $\mathbf{S}^{-1}$ and $\mathbf{S}^{1/2}$ denote the trace, determinant, inverse and square-root of $\mathbf{S}$, respectively. $\mathbf{S}\succeq \mathbf{0}$ and $\mathbf{S}\succ \mathbf{0}$ represent that $\mathbf{S}$ is positive semi-definite and positive definite, respectively. $\mathbb{C}^{M\times N}$ denotes the space of $M\times N$ complex matrices. $\|\mathbf{x}\|_2$ is the Euclidean norm of a complex vector $\mathbf{x}$. $\mathrm{Diag}(\mathbf{x})$ denotes a diagonal matrix with the main diagonal given by $\mathbf{x}$. For an arbitrary matrix $\mathbf{X}$, $\mathbf{X}^{T}$, $\mathbf{X}^{H}$ and $\mathrm{rank}\{\mathbf{X}\}$ represents the transpose, conjugate transpose and rank of $\mathbf{X}$, respectively. $\mathrm{vec}(\mathbf{X})$ denotes a column vector by stacking all the columns of $\mathbf{X}$. $\sim$ means ``distributed as''. $\mathcal{CN}(\mathbf{x},\mathbf{\Sigma})$ represents the circularly symmetric complex Gaussian random vector with mean $\mathbf{x}$ and covariance matrix $\mathbf{\Sigma}$.

\section{System Model and Problem Formulation}\label{S:systemModel}
We consider a downlink multi-cell cooperative network with  $K_t$ BSs, each equipped with $N_t$ antennas, as shown in Fig.~\ref{F:fig1}. Denote the total number of transmitting antennas as $M=K_tN_t$. At each time slot, $K_r$ MSs are scheduled and served by all the cooperating BSs. Each MS has $N_r$ antennas and thus can receive up to $N_r$ data streams.  Denote by $\mathbf{d}_k$ the information-bearing signal for the $k$th mobile station (denoted as $MS_k$), where $\mathbf{d}_k \in \mathbb{C}^{N_r \times 1}$. Assume Gaussian codebook is used and $\mathbf{d}_k\sim \mathcal{CN}(\mathbf{0},\mathbf{I}), \forall k$. Perfect channel state information (CSI) at the BSs is assumed and the precoding matrices for all the MSs are jointly determined. The total number of transmit antennas is assumed to be no less than the number of receiving antennas of the scheduled users, i.e., $M\geq K_rN_r$. In the sequel, we assume that $M= K_rN_r$ for simplicity. The received signal at $MS_k$ is then given by
\begin{align}
\mathbf{y}_k = \mathbf{H}_k \mathbf{W}_k \mathbf{d}_k+\sum\limits_{i=1,i\neq k}^{K_r}\mathbf{H}_k\mathbf{W}_i \mathbf{d}_i + \mathbf{n}_k,\quad k=1,\ldots,K_r
\end{align}
where $\mathbf{H}_k=\left[\mathbf{H}_{k1} \quad \mathbf{H}_{k2} \cdots \mathbf{H}_{kK_t}\right] \in \mathbb{C}^{N_r \times M}$ denotes the channel matrix for $MS_k$, which is assumed to be of full row rank. $\mathbf{H}_{kj} \in \mathbb{C}^{N_r \times N_t}$ is the channel from the $j$th base station (denoted as $BS_j$) to $MS_k$. $\mathbf{W}_k \in \mathbb{C}^{M \times N_r}$ is the precoding matrix for $MS_k$, with each column corresponding to one data stream. It is possible that the number of data streams for $MS_k$ is less than $N_r$, in which case, the corresponding columns of $\mathbf{W}_k$ are set to zero vectors. $\mathbf{n}_k\in \mathbb{C}^{N_r\times 1}$ denotes the receiver noise. Without loss of generality, we assume that $\mathbf{n}_k \sim \mathcal{CN}(\mathbf{0},\mathbf{I})$, $\forall k$. Under single-user decoding with multi-user interference treated as noise assumption, the achievable rate $R_k$ for $MS_k$ is given as \cite{209}
\begin{align}\label{E:rate}
R_k=\mathrm{log}\frac{\Big | \mathbf{I}+\sum\limits_{i=1}^{K_r}\mathbf{H}_k\mathbf{W}_i\mathbf{W}_i^H\mathbf{H}_k^H\Big |}{\Big | \mathbf{I}+\sum\limits_{i=1,i\neq k}^{K_r}\mathbf{H}_k\mathbf{W}_i\mathbf{W}_i^H\mathbf{H}_k^H\Big |},\quad k=1,\ldots,K_r
\end{align}

Denote the transmit covariance matrix for $MS_k$ as $\mathbf{S}_k=\mathbb{E}\left[\mathbf{W}_k\mathbf{d}_k\mathbf{d}_k^H\mathbf{W}_k^H\right]=\mathbf{W}_k\mathbf{W}_k^H$. Then $\mathbf{S}_k \in \mathbb{C}^{M \times M}$, $\mathbf{S}_k \succeq \mathbf{0}$ and $\text{rank}\{\mathbf{S}_k\} \leq N_r$. For $BS_j$, define a  binary matrix $\mathbf{B}_j$  as \cite{220}
\begin{align}
\mathbf{B}_j \triangleq \text{Diag}\big(\underbrace{0,\cdots,0}_{(j-1)N_t},\underbrace{1,\cdots,1}_{N_t},\underbrace{0,\cdots,0}_{(K_t-j)N_t}\big)
\end{align}
Without loss of generality, assume that all BSs have the same power constraints $P$. 
Then finding the optimal linear precoder for sum rate maximization under per-BS power constraints is equivalent to solving the following optimization problem
\begin{align}
\text{(P1):} \quad \underset{\{R_k\},\{\mathbf{S}_k\}}{\text{maximize}}\quad
&\sum_{k=1}^{K_r}R_k\\
\text{subject to}
\quad& R_k\leq\mathrm{log}\frac{\Big | \mathbf{I}+\sum\limits_{i=1}^{K_r}\mathbf{H}_k\mathbf{S}_i\mathbf{H}_k^H\Big |}{\Big | \mathbf{I}+\sum\limits_{i=1,i\neq k}^{K_r}\mathbf{H}_k\mathbf{S}_i\mathbf{H}_k^H\Big |},\quad \forall k \label{C:P1RateConstraint} \\
&\sum\limits_{k=1}^{K_r}\mathrm{Tr}\left(\mathbf{B}_j\mathbf{S}_k\right)\leq P, \quad\forall j \label{C:P1PowerConstraint}\\
&\mathbf{S}_k \succeq \mathbf{0},\quad \mathrm{rank}\{\mathbf{S}_k\} \leq N_r, \quad\forall k \label{C:P1RankConstraint}
\end{align}
where \eqref{C:P1PowerConstraint} represents the per-BS power constraints. Note that (P1) optimizes over the transmit covariance matrices $\{\mathbf{S}_k\}$ instead of the precoding matrices. The explicit rank constraint is necessary since otherwise, the ranks of the resulted transmit covariance matrices may exceed $N_r$, which is impractical due to the limited number of antennas at the receivers.
 (P1) is non-convex due to the non-convex rate and rank constraints. Therefore, it is difficult to find a global optimal solution efficiently. 

\section{BD with per-BS power constraints}\label{S:BD}
Under zero inter-user interference constraint, it has been shown that (P1) can be formulated into a convex optimization problem, from which the optimal BD solution can be efficiently obtained. This section reviews  BD  under per-BS power constraints, which is mainly based on \cite{220}.  BD completely eliminates the inter-user interference by ensuring that $\mathbf{H}_i\mathbf{W}_k=\mathbf{0}, \ \forall i\neq k$, or equivalently
\begin{align}\label{E:ZeroInterference}
\mathbf{H}_i\mathbf{S}_k\mathbf{H}_i^H = \mathbf{0}, \quad\forall i\neq k
\end{align}
Define $\mathbf{G}_k = \left[\mathbf{H}_1^T \ldots \mathbf{H}_{k-1}^T \quad \mathbf{H}_{k+1}^T \ldots \mathbf{H}_{K_r}^T\right]^T \in \mathbb{C}^{N_r(K_r-1)\times M}$. Perform singular value decomposition (SVD) to $\mathbf{G}_k$ to obtain
  \begin{align}
  \mathbf{G}_k=\mathbf{U}_k\left[\mathbf{\Sigma}_k\quad\mathbf{0} \right]\bigg[\begin{array}{c}
  \mathbf{V}_k^H \\
  \mathbf{\Tilde{V}}_k^H\end{array}\bigg],
  \end{align}
  where $\mathbf{U}_k\in\mathbb{C}^{N_r(K_r-1)\times N_r(K_r-1)}$, $\mathbf{V}_k\in\mathbb{C}^{M\times N_r(K_r-1)}$, $\mathbf{\Sigma}_k$ is a $N_r(K_r-1) \times N_r(K_r-1)$ positive diagonal matrix and $\mathbf{\Tilde{V}}_k\in \mathbb{C}^{M\times N_r}$ spans the null space of $\mathbf{G}_k$. Then \eqref{E:ZeroInterference} is satisfied by letting $\mathbf{S}_k=\Tilde{\mathbf{V}}_k\mathbf{Q}_k\Tilde{\mathbf{V}}_k^H$, where $\mathbf{Q}_k\in \mathbb{C}^{N_r\times N_r}$ and $\mathbf{Q}_k\succeq \mathbf{0}$ is the new design variable. With such a structure for $\mathbf{S}_k$, $\mathrm{rank}\{\mathbf{S}_k\}\leq N_r$ is automatically guaranteed.  Then finding the optimal BD to maximize the sum rate is equivalent to solving the following problem \cite{220}
  \begin{align}
\text{(P2):} \quad \underset{\{\mathbf{Q}_k\}}{\text{maximize}}\quad
&\sum_{k=1}^{K_r}\mathrm{log}\Big | \mathbf{I}+\mathbf{H}_k\Tilde{\mathbf{V}}_k\mathbf{Q}_k\Tilde{\mathbf{V}}_k^H\mathbf{H}_k^H\Big | \\
\text{subject to} \quad
&\sum\limits_{k=1}^{K_r}\mathrm{Tr}\left(\mathbf{B}_j\Tilde{\mathbf{V}}_k\mathbf{Q}_k\Tilde{\mathbf{V}}_k^H\right)\leq P, \quad\forall j \\
&\mathbf{Q}_k \succeq \mathbf{0}, \quad\forall k.
\end{align}
(P2) is convex, and hence can be solved efficiently with standard interior point method \cite{202} or existing software tools such as \texttt{CVX} \cite{227}. In \cite{220}, a closed form solution is derived.

\section{Improved Precoding over BD}\label{S:proposed}
BD performs very well in the high SNR regime and achieves the same DoF as the optimal linear/nonlinear precoding schemes \cite{220}.  However, in the low to medium SNR regime, the performance is poor. We therefore propose an extra step of optimization to improve the performance of BD in the low to medium SNR regime, yet preserve the good performance at high SNR.

Let $\{R_k^{BD}\}$ be the rate tuple achieved by BD. Consider the following optimization problem
\begin{align}
\text{(P3):}\quad\underset{\rho,\{\mathbf{S}_k\}}{\text{maximize}} \quad &-\rho\\
\label{C:P3RateConstraint}\text{subject to}\quad& \mathrm{log}\frac{\Big | \mathbf{I}+\sum\limits_{i=1}^{K_r}\mathbf{H}_k\mathbf{S}_i\mathbf{H}_k^H\Big |}{\Big | \mathbf{I}+\sum\limits_{i=1,i\neq k}^{K_r}\mathbf{H}_k\mathbf{S}_i\mathbf{H}_k^H\Big |}\geq R_k^{BD},\quad \forall k\\
&\sum\limits_{k=1}^{K_r}\mathrm{Tr}\left(\mathbf{B}_j\mathbf{S}_k\right)\leq \rho P, \quad\forall j\\
&\mathbf{S}_k \succeq \mathbf{0},\quad \mathrm{rank}\{\mathbf{S}_k\} \leq N_r, \quad\forall k \label{C:P3RankConstraint}
\end{align}
(P3) minimizes a common power factor $\rho$ for all BSs, while ensuring a minimum rate tuple  achieved by BD. Unlike  BD which completely eliminates inter-user interference, interference leakage is allowed in (P3). For the special case of $N_r=1$, (P3) can be transformed to the power minimization problem in \cite{226}, where an equivalent second order cone programming (SOCP) form is known. However, for the general case when $N_r\geq 2$, no convex formulation of (P3) is known. Before solving the problem, we will discuss how the solution to (P3) will help to find an improved precoder design over BD.
\begin{theorem}\label{T:OptPowerMinimization}
(P3) is guaranteed to be feasible and the solution $\{\rho^{opt},\{\mathbf{S}_k^{opt}\}\}$ satisfies $\rho^{opt}\leq 1$.
\end{theorem}
\begin{proof} It is easy to see that $\{\rho=1,\{\mathbf{S}_k^{BD}\}\}$ is feasible for (P3), where $\{\mathbf{S}_k^{BD}\}$ is the optimal BD transmit covariance matrices set. As a result, Theorem~\ref{T:OptPowerMinimization} follows.
\end{proof}

Although $\{\mathbf{S}_k^{opt}\}$ does not strictly increase the rates over $\{R_k^{BD}\}$\footnote{In fact, with $\{\mathbf{S}_k^{opt}\}$, the rate achieved by user $k$  equals to $R_k^{BD}$. This can be seen as follows. Suppose on the contrary, with the optimal solution $\{\rho^{opt},\{\mathbf{S}_k^{opt}\}\}$, there exists a user $k$ such that $R_k>R_k^{BD}$. Then we can strictly decrease the transmit power to user $k$ so that the minimum rate constraint is still satisfied. As a consequence, the power to other users can also be strictly decreased since the interference from user $k$ is reduced. This implies that the power factor $\rho$ can be further reduced, which contradicts that $\rho^{opt}$ is the optimal solution.},
the minimized power  factor $\rho^{opt}$ makes it possible to effectively suppress the noise and hence enhance the effective SNR. This can be achieved by using the new transmit covariance matrices $\mathbf{S}_k^{new}=\mathbf{S}_k^{opt}/\rho^{opt}, \forall k$. Since $\{\rho^{opt},\{\mathbf{S}_k^{opt}\}\}$ is feasible to (P3), it is easy to see that $\{\mathbf{S}_k^{new}\}$ satisfies the rank and power constraints in (P1), i.e., \ $\mathrm{rank}\{\mathbf{S}_k^{new}\}\leq N_r, \ \forall k$ and $\sum\limits_{k=1}^{K_r}\mathrm{Tr}\left(\mathbf{B}_j\mathbf{S}_k^{new}\right)\leq P, \ \forall j$.
Furthermore, the new achievable rate for $MS_k$ satisfies
\begin{align}
R_k^{new}&= \mathrm{log}\frac{\Big | \mathbf{I}+\frac{1}{\rho^{opt}}\sum_{i=1}^{K_r}{\mathbf{H}_k\mathbf{S}_i^{opt}\mathbf{H}_k^H} \Big |}
{\Big |\mathbf{I}+\frac{1}{\rho^{opt}}\sum_{i=1,i\neq k}^{K_r}{\mathbf{H}_k\mathbf{S}_i^{opt}\mathbf{H}_k^H} \Big |}\\\label{E:noiseSuppression}
&=\mathrm{log}\frac{\Big | \rho^{opt}\mathbf{I}+\sum_{i=1}^{K_r}{\mathbf{H}_k\mathbf{S}_i^{opt}\mathbf{H}_k^H} \Big |}
{\Big |\rho^{opt}\mathbf{I}+\sum_{i=1,i\neq k}^{K_r}{\mathbf{H}_k\mathbf{S}_i^{opt}\mathbf{H}_k^H} \Big |}\\
&\geq \mathrm{log}\frac{\Big | \mathbf{I}+\sum_{i=1}^{K_r}{\mathbf{H}_k\mathbf{S}_i^{opt}\mathbf{H}_k^H} \Big |}
{\Big |\mathbf{I}+\sum_{i=1,i\neq k}^{K_r}{\mathbf{H}_k\mathbf{S}_i^{opt}\mathbf{H}_k^H} \Big |}\\
&\geq R_k^{BD}
\end{align}
The last inequality follows since $\{\mathbf{S}_k^{opt}\}$ satisfies \eqref{C:P3RateConstraint}. The second last inequality follows since $\rho^{opt}\leq 1$.

The above relationship shows that the new set of transmit covariance matrices $\{\mathbf{S}_k^{new}\}$ will at least not decrease each user's rate over that achieved by BD.  With \eqref{E:noiseSuppression}, $R_k^{new}$ can be interpreted as the achievable rate by applying $\{\mathbf{S}_k^{opt}\}$ in an environment with noise power $\rho^{opt}$, instead of $1$ as in the original system. Since $\rho^{opt}\leq 1$, this implies an effective SNR enhancement by $10\mathrm{log}_{10}\left(1/\rho^{opt}\right)$ dB for $\{\mathbf{S}_k^{new}\}$ over $\{\mathbf{S}_k^{opt}\}$. Furthermore, since $\{\mathbf{S}_k^{opt}\}$ performs at least as well as $\{\mathbf{S}_k^{BD}\}$ due to \eqref{C:P3RateConstraint}, then with $\{\mathbf{S}_k^{new}\}$, an effective SNR enhancement by $10\mathrm{log}_{10}\left(1/\rho^{opt}\right)$ dB over BD is also guaranteed. Such SNR enhancement provides a way to compensate the transmit power boost problem associated with BD, and hence increase the achievable rate. We are now ready to present the algorithms to solve (P3).
\subsection{Solve (P3) When $N_r=1$}
When each MS has single antenna, and hence single data stream only, BD reduces to the well-known zero-forcing (ZF) precoding \cite{174,221}. Denote the channel vector to $MS_k$ as $\mathbf{h}_k\in \mathbb{C}^{1\times M}$, then (P3) can be equivalently formulated into the following problem \cite{226}
 \begin{align}
\text{(P4):}\quad \underset{\rho,\{\mathbf{w}_k\}}{\text{minimize}} \quad & \rho\\
 \text{subject to} \quad & \frac{|\mathbf{h}_k\mathbf{w}_k|^{2}}{1+\sum_{i=1,i\neq k}^{K_r}|\mathbf{h}_k\mathbf{w}_i|^2} \geq \gamma_k^{ZF},\quad \forall k\\
 &\sum_{k=1}^{K_r}\|\mathbf{w}_k^{[j]}\|^2 \leq \rho P, \quad \forall j
 \end{align}
 where $\{\gamma_k^{ZF}\}$ is the SINR tuple achieved with the ZF precoding, $\mathbf{w}_k\in \mathbb{C}^{M\times 1}$ is the precoding vector for $MS_k$ and $\mathbf{w}_k^{[j]}\in \mathbb{C}^{N_t\times 1}$ corresponds to the precoding vector for $MS_k$ used by $BS_j$. The above problem can be transformed into an equivalent SOCP as follows \cite{226}
 \begin{equation}
 \begin{aligned}
 \underset{\Tilde{\rho},\{\mathbf{w}_k\}}{\text{minimize}}\quad & \Tilde{\rho}\\
 \text{subject to} \quad & \left\|
 \begin{array}{c}
 \left[\mathbf{h}_k\mathbf{W}\right]^T \\
 1\\
 \mathbf{h}_k\mathbf{w}_k
 \end{array}
 \right\|_2 \preceq \mathcal{K}\mathbf{0},\quad \forall k\\
 &\left\|
 \begin{array}{c}
 \mathrm{vec}\left(\mathbf{M}_j\right)\\
 \Tilde{\rho}\sqrt{P}
 \end{array}
 \right\|_2\preceq \mathcal{K}\mathbf{0},\quad\forall j
 \end{aligned}
 \end{equation}
 where $\rho=\Tilde{\rho}^2$, $\mathbf{W}\triangleq\left[\mathbf{w}_1\quad\mathbf{w}_2\ldots\mathbf{w}_{K_r}\right]\in\mathbb{C}^{M\times K_r}$, $\mathbf{M}_j\triangleq \left[\mathbf{w}_1^{[j]}\quad\mathbf{w}_2^{[j]}\ldots\mathbf{w}_{K_r}^{[j]} \right]\in\mathbb{C}^{N_t\times K_r}$. For any vector $\mathbf{y}\in \mathbb{C}^{n \times 1}, x\in \mathbb{R}$, $\bigg\|\begin{array}{c}\mathbf{y}\\x\end{array} \bigg\|_2\preceq \mathcal{K}\mathbf{0}$ represents the second order cone constraint $\|\mathbf{y}\|_2^2\leq x^2$. The SOCP is convex and can be solved efficiently with software tools such as \texttt{CVX} \cite{227}.

 \subsection{Solve (P3) When $N_r\geq 2$}
 When $N_r \geq 2$, no convex formulation for (P3) is known. The non-convexity arises from the non-convex rate and rank constraints \eqref{C:P3RateConstraint} and \eqref{C:P3RankConstraint}. In this subsection, we propose an efficient algorithm to solve (P3) approximately. Firstly, the rate constraints \eqref{C:P3RateConstraint} are convexified by applying the following first-order Taylor approximation\footnote{Note that although \eqref{E:rateApprox1} is sufficient to convexify \eqref{C:P3RateConstraint}, \eqref{E:rateApprox2} is necessary to handle the non-convex rank constraints given by \eqref{C:P3RankConstraint}.}
  \begin{align}\label{E:rateApprox1}
  \mathrm{log}\Big |\mathbf{I}+\mathbf{H}_k\Big(\sum_{i\neq k}^{K_r}\mathbf{S}_i\Big)\mathbf{H}_k^H \Big |\approx \mathrm{Tr}\Big[\mathbf{H}_k\Big(\sum_{i\neq k}^{K_r}\mathbf{S}_i\Big)\mathbf{H}_k^H\Big]
  =\mathrm{Tr}\Big[\mathbf{H}_k^H\mathbf{H}_k\Big(\sum_{i\neq k}^{K_r}\mathbf{S}_i\Big)\Big]
  \end{align}
  \begin{equation}\label{E:rateApprox2}
  \begin{aligned}
  \mathrm{log}\Big|\mathbf{I}+ \mathbf{H}_k\Big(\sum_{i=1}^{K_r}\mathbf{S}_i\Big)\mathbf{H}_k^H\Big|&\approx
  \mathrm{log}\big|\mathbf{I}+\mathbf{H}_k\mathbf{S}_k\mathbf{H}_k^H\big|
  +\mathrm{Tr}\Big[\big(\mathbf{I}+\mathbf{H}_k\mathbf{S}_k\mathbf{H}_k^H\big)^{-1}\mathbf{H}_k\Big(\sum_{i\neq k}^{K_r}\mathbf{S}_i\Big)\mathbf{H}_k^H\Big]\\
  &\approx \mathrm{log}\big|\mathbf{I}+\mathbf{H}_k\mathbf{S}_k\mathbf{H}_k^H\big|
  +\mathrm{Tr}\Big[\mathbf{H}_k^H\big(\mathbf{I}+\mathbf{H}_k\mathbf{S}_k^{BD}\mathbf{H}_k^H\big)^{-1}\mathbf{H}_k\Big(\sum_{i\neq k}^{K_r}\mathbf{S}_i\Big)\Big]
  \end{aligned}
  \end{equation}
  where the identity $\mathrm{Tr}(\mathbf{A}\mathbf{B})=\mathrm{Tr}(\mathbf{B}\mathbf{A})$ has been used. In \eqref{E:rateApprox2}, the gradient of the log-determinant function at the point $\mathbf{I}+\mathbf{H}_k\mathbf{S}_k\mathbf{H}_k^H$ has been approximated as $\big(\mathbf{I}+\mathbf{H}_k\mathbf{S}_k^{BD}\mathbf{H}_k^H\big)^{-1}$. With \eqref{E:rateApprox1} and \eqref{E:rateApprox2}, (P3) can be approximated as
  \begin{align}
  \text{(P5):}\quad\underset{\rho,\{\mathbf{S}_k\}}{\text{maximize}}\quad &-\rho\\
  \text{subject to}\quad & \mathrm{log}\big|\mathbf{I}+\mathbf{H}_k\mathbf{S}_k\mathbf{H}_k^H \big|\geq \mathrm{Tr}\Big(\mathbf{F}_k\sum_{i\neq k}^{K_r}\mathbf{S}_i\Big)+R_k^{BD},\quad \forall k \label{C:P5RateConstraint} \\
  &\sum_{k=1}^{K_r}\mathrm{Tr}\big(\mathbf{B}_j\mathbf{S}_k\big) \leq \rho P, \quad \forall j \label{C:P5PowerConstraint} \\
  &\mathbf{S}_k \succeq \mathbf{0},\quad \mathrm{rank}\{\mathbf{S}_k\}\leq N_r, \quad \forall k
  \end{align}
  where
  $ \mathbf{F}_k=\mathbf{H}_k^H\big[\mathbf{I}-(\mathbf{I}+\mathbf{H}_k\mathbf{S}_k^{BD}\mathbf{H}_k^H)^{-1}\big]\mathbf{H}_k$.

   It can be verified that $\{\rho=1,\{\mathbf{S}_k^{BD}\}\}$ is still feasible for (P5), so the solution $\{\rho^{\star\star}, \{\mathbf{S}_k^{\star\star}\}\}$ to  (P5) still satisfies $\rho^{\star\star}\leq 1$. Due to the rank constraint, (P5) is still non-convex. However, by solving the rank-relaxed problem (denoted by (R-P5)) with the dual method, we show that the optimal solution is guaranteed to satisfy the rank constraint, and hence it is also an optimal solution of the non-convex problem (P5).
  Denote by $\{\lambda_k\}$ and $\{\mu_j\}$ the set of dual variables of (R-P5), which are associated with the rate \eqref{C:P5RateConstraint} and per-BS power constraints \eqref{C:P5PowerConstraint}, respectively. Then the Lagrangian function of (R-P5) can be written as
  \begin{align}
  L&\left(\rho,\{\mathbf{S}_k\},\{\lambda_k\},\{\mu_j\}\right) \notag \\
  =&-\rho+\sum_{k=1}^{K_r}\lambda_k\Big[\mathrm{log}|\mathbf{I}+\mathbf{H}_k\mathbf{S}_k\mathbf{H}_k^H|-\mathrm{Tr}\big(\mathbf{F}_k\sum_{i\neq k}^{K_r}\mathbf{S}_i\big)-R_k^{BD}\Big]
  +\sum_{j=1}^{K_t}\mu_j\Big[\rho P-\sum_{k=1}^{K_r}\mathrm{Tr}\big(\mathbf{B}_j\mathbf{S}_k\big) \Big]  \notag \\
  =&\rho\Big(P\sum_{j=1}^{K_t}\mu_j-1\Big)+\sum_{k=1}^{K_r}\Big[\lambda_k\mathrm{log}|\mathbf{I}+\mathbf{H}_k\mathbf{S}_k\mathbf{H}_k^H|-\mathrm{Tr}(\mathbf{C}_k\mathbf{S}_k)-\lambda_k R_k^{BD}\Big], \label{Lagrangian}
  \end{align}
where $\mathbf{C}_k\triangleq \displaystyle\sum\limits_{j=1}^{K_t}\mu_j\mathbf{B}_j+\sum_{i=1,i\neq k}^{K_r}\lambda_i\mathbf{F}_i$. The Lagrangian dual objective is then written as
\begin{equation}\label{Primal}
\begin{aligned}
g\left(\{\lambda_k\},\{\mu_j\}\right)&=\underset{\mathbf{S}_k\succeq\mathbf{0},\forall k}{\text{max}}\ \underset{\rho}{\text{max}}\ L\left(\rho,\{\mathbf{S}_k\},\{\lambda_k\},\{\mu_j\}\right)\\
&=\begin{cases}
\underset{\mathbf{S}_k\succeq \mathbf{0},\forall{k}}{\text{max}}\ \Tilde{L}\left(\{\mathbf{S}_k\},\{\lambda_k\},\{\mu_j\}\right),\quad & \text{if $\sum_{j=1}^{K_t}\mu_j=1/P$}\\
\infty,\quad & \text{otherwise}
\end{cases}
\end{aligned}
\end{equation}
where $\Tilde{L}\left(\{\mathbf{S}_k\},\{\lambda_k\},\{\mu_j\}\right)\triangleq \sum_{k=1}^{K_r}\Big[\lambda_k\mathrm{log}|\mathbf{I}+\mathbf{H}_k\mathbf{S}_k\mathbf{H}_k^H|-\mathrm{Tr}(\mathbf{C}_k\mathbf{S}_k)-\lambda_k R_k^{BD}\Big]$.

Note that since $L\left(\rho,\{\mathbf{S}_k\},\{\lambda_k\},\{\mu_j\}\right)$ is an affine function of $\rho$, $g\left(\{\lambda_k\},\{\mu_j\}\right)$ is finite only when $\sum_{j=1}^{K_t}\mu_j=1/P$. Since the dual variables should be chosen such that the Lagrangian dual function is bounded, this imposes equality constraints on the dual optimization problem of (R-P5), which is stated as
\begin{equation}
\begin{aligned}
\text{(R-P5-D)}:\quad\underset{\lambda_k\geq 0,\mu_j\geq 0,\forall{k,j}}{\text{minimize}}\quad & g\left(\{\lambda_k\},\{\mu_j\}\right)\\
\text{subject to} \quad & \sum_{j=1}^{K_t}\mu_j=1/P\\
\end{aligned}
\end{equation}
Since (R-P5) is convex and satisfies the Slater's condition \cite{202}, the duality gap between the optimal objective function value of (R-P5) and that of its dual (R-P5-D) is zero. Thus, the optimal solution can be obtained by simultaneously updating the primal variables, $\{\mathbf{S}_k\}$ and the dual variables $\{\lambda_k\}$ and $\{\mu_j\}$. For a given set of dual variables $\{\lambda_k\}$ and $\{\mu_j\}$, $\{\mathbf{S}_k^{\star}\}$ can be updated by solving the maximization problem \eqref{Primal}. With $\{\mathbf{S}_k^{\star}\}$, the dual variables $\{\lambda_k\}$ and $\{\mu_j\}$ can be updated with subgradient-based method \cite{229}.

\subsubsection{Primal Update}
We firstly focus on solving for $\{\mathbf{S}_k^{\star}\}$ with a given set of dual variables $\{\lambda_k\}$ and $\{\mu_j\}$. It can be observed from \eqref{Primal} that the maximization of $ \Tilde{L}\left(\{\mathbf{S}_k\},\{\lambda_k\},\{\mu_j\}\right)$ over $\{\mathbf{S}_k\}$ can be decoupled into $K_r$ parallel sub-problems, each solving for one $\mathbf{S}_k$. By discarding the irrelevant terms, the subproblem for solving $\mathbf{S}_k$ , given $\{\lambda_k\}$ and $\{\mu_j\}$, is
\begin{align}
\text{(P6):}\quad\underset {\mathbf{S}_k\succeq \mathbf{0}}{\text{maximize}} \ \lambda_k\mathrm{log}|\mathbf{I}+\mathbf{H}_k\mathbf{S}_k\mathbf{H}_k^H|-\mathrm{Tr}(\mathbf{C}_k\mathbf{S}_k),
\end{align}
where $\mathbf{C}_k\triangleq \displaystyle\sum\limits_{j=1}^{K_t}\mu_j\mathbf{B}_j+\sum_{i=1,i\neq k}^{K_r}\lambda_i\mathbf{F}_i\in \mathbb{C}^{M\times M}$.
\begin{lemma}\label{Lemma:rankCk}
For (P6) to have a bounded objective value, the dual variables $\{\lambda_k\}$ and $\{\mu_j\}$ should have values such that $\mathbf{C}_k$ is positive definite, i.e., $\mathbf{C}_k\succ \mathbf{0}$
\end{lemma}
\begin{proof} See Appendix~\ref{ProofRankCk}.\end{proof}

 With Lemma~\ref{Lemma:rankCk}, $\mathbf{C}_k$ can be decomposed as $\mathbf{C}_k=\mathbf{C}_k^{1/2}\mathbf{C}_k^{1/2}$, where $\mathbf{C}_k^{1/2}$ is  Hermitian and invertible. Furthermore, $\mathrm{Tr}(\mathbf{C}_k\mathbf{S}_k)=\mathrm{Tr}(\mathbf{C}_k^{1/2}\mathbf{S}_k\mathbf{C}_k^{1/2})$. Define $\Tilde{\mathbf{S}}_k\triangleq \mathbf{C}_k^{1/2}\mathbf{S}_k\mathbf{C}_k^{1/2}$, then $\mathbf{S}_k=\mathbf{C}_k^{-1/2}\Tilde{\mathbf{S}}_k\mathbf{C}_k^{-1/2}$. Then (P6) is equivalent to
\begin{align}\label{E:TildeSk}
\underset{\Tilde{\mathbf{S}}_k\succeq \mathbf{0}}{\text{maximize}}\ \lambda_k\mathrm{log}|\mathbf{I}+\mathbf{H}_k\mathbf{C}_k^{-1/2}\Tilde{\mathbf{S}}_k\mathbf{C}_k^{-1/2}\mathbf{H}_k^H|-\mathrm{Tr}(\Tilde{\mathbf{S}}_k)
\end{align}
To find the optimal $\Tilde{\mathbf{S}}_k$, express the (reduced) SVD of $\mathbf{H}_k\mathbf{C}_k^{-1/2}\in\mathbb{C}^{N_r\times M}$ as
\begin{align}
\mathbf{H}_k\mathbf{C}_k^{-1/2}=\bar{\mathbf{U}}_k\bar{\mathbf{\Sigma}}_k\bar{\mathbf{V}}_k^H,
\end{align}
where $\bar{\mathbf{U}}_k\in\mathbb{C}^{N_r\times N_r}$,$\bar{\mathbf{V}}_k\in\mathbb{C}^{M\times N_r}$ and $\bar{\mathbf{U}}_k^H\bar{\mathbf{U}}_k=\bar{\mathbf{V}}_k^H\bar{\mathbf{V}}_k=\mathbf{I}_{N_r}$. $\bar{\mathbf{\Sigma}}_k=\mathrm{Diag}(\bar{\sigma}_{k,1},\ldots,\bar{\sigma}_{k,N_r})$. Then \eqref{E:TildeSk} is equivalent to
\begin{align}\label{E:TildeSk2}
\underset{\Tilde{\mathbf{S}}_k\succeq \mathbf{0}}{\text{maximize}}\ \lambda_k \mathrm{log}|\mathbf{I}+\bar{\mathbf{V}}_k^H\Tilde{\mathbf{S}}_k\bar{\mathbf{V}}_k\bar{\mathbf{\Sigma}}_k^2|-\mathrm{Tr}(\Tilde{\mathbf{S}}_k)
\end{align}
Applying the Hadamard's inequality \cite{209}, the optimal solution to \eqref{E:TildeSk2} and hence to \eqref{E:TildeSk} is given as $\Tilde{\mathbf{S}}_k^{\star}=\bar{\mathbf{V}}_k\mathbf{D}_k\bar{\mathbf{V}}_k^H$, where $\mathbf{D}_k=\mathrm{Diag}(d_{k,1},\ldots,d_{k,N_r})$ with $d_{k,s}$  obtained by standard water-filling algorithm \cite{209}
\begin{align}
d_{k,s}=\bigg(\lambda_k-\frac{1}{\bar{\sigma}_{k,s}^2} \bigg)^+, \quad s=1,\ldots,N_r,
\end{align}
where $(x)^+\triangleq \max(0,x)$. With such results, the optimal solution to (P6) for a given set of dual variables $\{\lambda_k\}$ and $\{\mu_j\}$ is given as
\begin{align}\label{E:SkStarSolution}
\mathbf{S}_k^{\star}=\mathbf{C}_k^{-1/2}\bar{\mathbf{V}}_k\mathbf{D}_k\bar{\mathbf{V}}_k^H\mathbf{C}_k^{-1/2}
\end{align}
When the optimal solution for dual variables  $\{\lambda_k\}$ and $\{\mu_j\}$ is obtained, the corresponding solution in \eqref{E:SkStarSolution} (now denoted by $\mathbf{S}_k^{\star\star}$) becomes optimal for (R-P5).
\begin{remark}
Since $\bar{\mathbf{V}}_k\in\mathbb{C}^{M\times N_r}$, $\mathrm{rank}\{\mathbf{S}_k^{\star\star}\}\leq N_r$ is automatically satisfied due to \eqref{E:SkStarSolution}. As a result, $\{\mathbf{S}_k^{\star\star}\}$ is an optimal solution to the rank constrained non-convex problem (P5) as well. On the other hand, if (R-P5) is directly solved with software tools such as \texttt{CVX} \cite{227}, there is no guarantee that the rank constraints will be satisfied.
\end{remark}
\begin{remark}
With $\{\mathbf{S}_k^{\star\star}\}$ obtained, the optimal power factor to (P5) can be calculated as
\begin{align}\label{E:rhoStarStar}
\rho^{\star\star}=\underset{j\in\{1,\ldots,K_t\}}{\max}\ \frac{1}{P}\sum_{k=1}^{K_r}\mathrm{Tr}(\mathbf{B}_j\mathbf{S}_k^{\star\star})
\end{align}
\end{remark}

\subsubsection{Dual Update}
We now focus on solving the dual problem (R-P5-D). The dual variables $\{\lambda_k\}$ and $\{\mu_j\}$ can be updated with subgradient-based method after finding $\{\mathbf{S}_k^{\star}\}$. The equality constraint in (R-P5-D) can be eliminated by substituting $\mu_{K_t}=\frac{1}{P}-\sum_{j=1}^{K_t-1}\mu_j$ so that the problem dimension is reduced by $1$. Then the dual function after substitution of $\mu_{K_t}$ is given by
\begin{equation}\label{ReducedLagrangian}
\begin{aligned}
\Tilde{g}\left(\{\lambda_k\},\{\mu_j\}_{j=1}^{K_t-1}\right)=\underset{\mathbf{S}_k\succeq \mathbf{0},\forall k}{\max}
& \Bigg\{  \sum\limits_{k=1}^{K_r}\lambda_k \big[\mathrm{log}|\mathbf{I}+\mathbf{H}_k\mathbf{S}_k\mathbf{H}_k^H|-\mathrm{Tr}(\mathbf{F}_k\sum_{i\neq k}^{K_r}\mathbf{S}_i)-R_k^{BD}\big]  \\
&   +\sum_{j=1}^{K_t-1}\mu_j \sum_{k=1}^{K_r}\mathrm{Tr}\big[(\mathbf{B}_{K_t}-\mathbf{B}_j)\mathbf{S}_k\big]-\frac{1}{P}\sum_{k=1}^{K_r}\mathrm{Tr}(\mathbf{B}_{K_t}\mathbf{S}_k) \Bigg\}
\end{aligned}
\end{equation}
 Then (R-P5-D) is equivalent to
\begin{equation}
\begin{aligned}
\text{(P7):}\quad\underset{\{\lambda_k\},\{\mu_j\}_{j=1}^{K_t-1}}{\text{minimize}}\quad &  \Tilde{g}\left(\{\lambda_k\},\{\mu_j\}_{j=1}^{K_t-1}\right)\\
\text{subject to} \quad & \sum_{j=1}^{K_t-1}\mu_j\leq 1/P\\
& \mu_j\geq 0,\quad j=1,\ldots,K_t-1\\
& \lambda_k\geq 0,\quad k=1,\ldots,K_r
\end{aligned}
\end{equation}
The subgradient of (P7) can be found with the following Lemma.
\begin{lemma}\label{Lemma:subgradients}
 With the primal solution $\{\mathbf{S}_k^{\star}\}$ given by \eqref{E:SkStarSolution} for a given set of dual variables $\{\lambda_k\}$ and $\{\mu_j\}$, the subgradient of
$\Tilde{g}\left(\{\lambda_k\},\{\mu_j\}_{j=1}^{K_t-1}\right)$ is given by
\begin{align}
s_{\lambda_k}&=\mathrm{log}|\mathbf{I}+\mathbf{H}_k\mathbf{S}_k^{\star}\mathbf{H}_k^H|-\mathrm{Tr}(\mathbf{F}_k\sum_{i\neq k}^{K_r}\mathbf{S}_i^{\star})-R_{k}^{BD},\quad k=1,\ldots,K_r \label{E:subgradientLambda} \\
s_{\mu_j}&=\sum_{k=1}^{K_r}\mathrm{Tr}[(\mathbf{B}_{K_t}-\mathbf{B}_j)\mathbf{S}_k^{\star}],\quad j=1,\ldots,K_t-1 \label{E:subgradientMu}
\end{align}
\end{lemma}
\begin{proof} See Appendix~\ref{ProofSubgradientLemma}.\end{proof}

With the subgradient obtained, the dual variables can then be updated with subgradient-based method, such as ellipsoid method \cite{228}.

\subsubsection{Primal-dual Method for (P5)}
The algorithm for solving (R-P5), and hence the non-convex problem (P5) is now summarized in Algorithm~\ref{algo1}.
\begin{algorithm}[H]
\caption{Primal-dual Method for (P5)}
\begin{algorithmic}[1]\label{algo1}
\STATE Initialize $\lambda_k\geq 0, \forall k$ and $\mu_j\geq 0$, $j\in \{1,\ldots,K_t-1\}$,\ $\sum_{j=1}^{K_t-1}\mu_j\leq 1/P$.
\REPEAT
\STATE With $\{\lambda_k\}$ and $\{\mu_j\}_{j=1}^{K_t}$, where $\mu_{K_t}=1/P-\sum_{j=1}^{K_t-1}\mu_j$, solve for $\{\mathbf{S}_k^{\star}\}$ using \eqref{E:SkStarSolution}.
\STATE Compute the subgradient of $\Tilde{g}\left(\{\lambda_k\},\{\mu_j\}_{j=1}^{K_t-1}\right)$ using \eqref{E:subgradientLambda} and \eqref{E:subgradientMu}, then update $\{\lambda_k\}$ and $\{\mu_j\}_{j=1}^{K_t-1}$ accordingly based on the ellipsoid method \cite{228}.
\UNTIL{$\{\lambda_k\}_{k=1}^{K_r}$ and $\{\mu_j\}_{j=1}^{K_t-1}$ converge to a prescribed accuracy.}
\STATE Then $\{\mathbf{S}_k^{\star}\}$ approaches to the optimal solution $\{\mathbf{S}_k^{\star\star}\}$. Set $\rho^{\star\star}$ using \eqref{E:rhoStarStar}.
\end{algorithmic}
\end{algorithm}


 \subsection{Improved precoding over BD}
 Based on previous discussions, for the given optimal BD solution (or ZF precoding when $N_r=1$), the following steps can be applied to find an improved linear precoder design.
 \begin{algorithm}[H]
\caption{Improved precoding over BD}
\begin{algorithmic}[1]
\STATE Solve (P5) with Algorithm~\ref{algo1} $\big($or (P4) with \texttt{CVX} when $N_r=1 \big)$. Denote the solution as $\big\{\{\mathbf{S}_k^{\star\star}\},\rho^{\star\star}\big\}$ $\big($or $\big\{\{\mathbf{w}_k^{\star}\}, \rho^{\star}\big\}$ when $N_r=1$$\big)$.
\STATE Set the proposed transmit covariance matrices as $\mathbf{S}_k^{prop}=\mathbf{S}_k^{\star\star}/\rho^{\star\star}, $ $\big($or the proposed precoder when $N_r=1$ as $\mathbf{w}_k^{prop}=\mathbf{w}_k^{\star}/\sqrt{\rho^{\star}}$$\big), \forall k$.
\end{algorithmic}
\end{algorithm}

\section{Numerical results}\label{S:simulation}
 This section presents the numerical results. For the simulations below, the entries of the channel matrices $\mathbf{H}_k$ are independently and identically distributed (i.i.d) circularly symmetric complex Gaussian random variables with zero mean and unit variance. Since the noise power is normalized, the system SNR is defined as $SNR\triangleq P$, where $P$ is the maximum power for each BS. Algorithm~\ref{algo1} is terminated when the volume of the ellipsoid containing the optimal dual variables is sufficiently small, or more specifically, when $\sqrt{\mathbf{s}^{T}\mathbf{E}\mathbf{s}}\leq 10^{-6}$, where $\mathbf{s}$ is the subgradient vector, $\mathbf{E}$ is the positive definite matrix whose eigenvectors define the principal  directions of the ellipsoid.

\subsection{Convergence Behavior of Algorithm~\ref{algo1}}
The convergence behavior of Algorithm $1$ is illustrated with one channel realization at $SNR=0$ dB. A network with $[K_t \ N_t \ K_r \ N_r]=[3 \ 2 \ 3 \ 2]$ is simulated. The initial values of the dual variables are assigned with $\mu_j=1/(PK_t), \forall j$ and $\lambda_k=0.1, \forall k$. Algorithm \ref{algo1} generates a sequence of transmit covariance matrices set $\{\mathbf{S}_k^{\star}\}$. The achievable sum rate with the scaled covariance matrices $\{\mathbf{S}_k^{\star}/\rho^{\star}\}$ is plotted in Fig.~\ref{F:fig3}, where similar to \eqref{E:rhoStarStar}, $\rho^{\star}= \underset{j\in\{1,\cdots,K_t\}}{\max}(1/P)\sum_{k=1}^{K_r}\mathrm{Tr}(\mathbf{B}_j\mathbf{S}_k^{\star})$. Such scaling ensures that the per-BS power constraints are satisfied. The BD solution is also plotted with dotted line for comparison. It is observed that the algorithm eventually converges to a fixed sum rate, which significantly outperforms the optimal BD. Similar to that in \cite{220}, the convergence speed depends on the total number of dual variables, $K_r+K_t-1$. With the ellipsoid method, it is known that the complexity is of the order $\mathcal{O}[(K_r+K_t-1)^2]$ for large system. It is noted that the convergence point does not necessarily give the optimal solution, since higher sum rate has been observed in previous iterations. This is due to the approximations that have been made for solving (P3). However, the algorithm does converge to a  point with a sum rate very close to the highest rate that has appeared so far, as shown in Fig.~\ref{F:fig3}.
\subsection{Sum Rate Comparison}
The sum rate achieved with the proposed scheme is compared with the optimal BD, as well as two MMSE-based precoding schemes \cite{224,225}, denoted as ``MMSE Zhang'' and ``MMSE Shi'' in the figure, respectively. A network with parameters $[K_t \ N_t \ K_r \ N_r]=[3 \ 2 \ 3 \ 2]$ is simulated. The average sum rate over $10000$ channel realizations is plotted in Fig~\ref{F:fig4}. Firstly, it is observed that the two MMSE-based schemes, although provide some rate gain over BD at low SNR, perform worse than BD in the high SNR regime. Furthermore, the performance degradation increases with SNR. On the other hand, the proposed scheme outperforms the optimal BD across all SNR ranges and the gain is more pronounced in the low to medium SNR regime. The average value of $1/\rho^{\star\star}$ in dB, with $\rho^{\star\star}$ the optimal power factor for (P5), is also plotted in Fig.~\ref{F:fig5}. Since (P5) is an approximated problem formulation of (P3), $1/\rho^{\star\star}$ in dB can be viewed as the approximated SNR enhancement, as discussed in Section~\ref{S:proposed}.  Fig.~\ref{F:fig5} verifies the sum rate gain in Fig.~\ref{F:fig4} and it also shows that solving the non-convex problem (P3) by solving (P5) is a reasonable approximation.

\section{Conclusions}\label{S:conclusions}
 This paper proposes an improved linear precoding scheme over over BD in multi-cell cooperative downlink networks under per-BS power constraints. The performance gain is achieved by applying an effective-SNR-enhancement technique. It is shown that by solving a power minimization problem subject to a minimum rate constraint achieved by BD, and using the properly scaled transmit covariance matrices at each transmitter, the system noise can be effectively suppressed and the SNR can be enhanced. Such a technique provides a method to compensate the transmit power boost problem associated with BD. The power minimization problem is in general non-convex, due to the non-convex rate and rank constraints. In order to find an efficient solution, the rate constraint is convexified by using Taylor approximation. Then the rank-relaxed convexified problem is solved with the dual method. The closed form solution shows that there is always an optimal solution for the rank-relaxed problem such that the rank constraint is guaranteed to be satisfied. Therefore, the solution is also optimal to the rank-constrained non-convex problem. The proposed scheme is efficient since only convex optimization problem is required to be solved. Simulation results show a significant sum rate gain over the optimal BD and existing MMSE-based schemes.

\appendices
\section{Proof of Lemma~\ref{Lemma:rankCk}}\label{ProofRankCk}
It can be verified that at the optimal solution to (R-P5), the inequality constraints \eqref{C:P5RateConstraint} will be active. Then based on the complementary slackness condition \cite{202}, we can assume that the optimal dual variables  $\{\lambda_k^{\star}\}$ are positive. Therefore, we can assume that $\lambda_k>0$ in (P6). We then prove Lemma~\ref{Lemma:rankCk} by contradiction. Since $\mathbf{C}_k$ is Hermitian, all the eigenvalues are real. Suppose that $\mathbf{C}_k$ has a non-positive eigenvalue, i.e., $\exists \alpha\leq 0 $ and a normalized vector $\mathbf{q}$, with $\mathbf{q}^H\mathbf{q}=1$ such that $\mathbf{C}_k\mathbf{q}=\alpha \mathbf{q}$. Then let $\mathbf{S}_k=t\mathbf{q}\mathbf{q}^{H}$ with $t\geq 0$. Substituting into the objective function of (P6) yields
\begin{align}
\quad & \lambda_k\mathrm{log}\left | \mathbf{I}+t\mathbf{H}_k\mathbf{q}\mathbf{q}^H\mathbf{H}_k^H \right |-\mathrm{Tr}(t\mathbf{C}_k\mathbf{q}\mathbf{q}^H) \notag \\
=&\lambda_k\mathrm{log}(1+t \|\mathbf{H}_k\mathbf{q}\|^2)-\alpha t. \label{P6Value}
\end{align}
 Since $\alpha \leq 0$, as $t\rightarrow \infty$, the value of \eqref{P6Value} becomes unbounded provided that $\mathbf{H}_k\mathbf{q}\neq \mathbf{0}$ (which is true with probability one with independent channel realizations and the fact that $\mathbf{C}_k$ does not depend on $\mathbf{H}_k$). Therefore, we conclude that in order to have a bounded objective value for (P6), all eigenvalues of $\mathbf{C}_k$ should be positive. As a result, Lemma~\ref{Lemma:rankCk} follows.

\section{Proof of Lemma~\ref{Lemma:subgradients}}\label{ProofSubgradientLemma}
A vector $\mathbf{s}$ is a subgradient of function $\Tilde{g}(\mathbf{x})$ at point $\mathbf{x}$ if
\begin{align}
\Tilde{g}\left(\bar{\mathbf{x}}\right)\geq \Tilde{g}\left(\mathbf{x}\right)+\mathbf{s}^{T}(\bar{\mathbf{x}}-\mathbf{x}), \forall \bar{\mathbf{x}},
\end{align}
Or equivalently, the vector formed by $\{s_{\lambda_k}\}$ and $\{s_{\mu_j}\}$ is a subgradient of $\Tilde{g}\left(\{\lambda_k\},\{\mu_j\}_{j=1}^{K_t-1}\right)$ if
\begin{align}\label{E:subgDef}
\Tilde{g}\left(\{\bar{\lambda}_k\},\{\bar{\mu}_j\}_{j=1}^{K_t-1}\right)
\geq \Tilde{g}\left(\{\lambda_k\},\{\mu_j\}_{j=1}^{K_t-1}\right) + \sum_{k=1}^{K_r}s_{\lambda_k}(\bar{\lambda}_k-\lambda_k)
+\sum_{j=1}^{K_t-1}s_{\mu_j}(\bar{\mu}_j-\mu_j), \forall \bar{\lambda}_k,\bar{\mu}_j
\end{align}
For the given $\{\lambda_k\}$ and $\{\mu_j\}_{j=1}^{K_t-1}$, denote by $\{\mathbf{S}_k^{\star}\}$ be the transmit covariance matrices that achieves the maximum dual function value $\Tilde{g}\left(\{\lambda_k\},\{\mu_j\}_{j=1}^{K_t-1}\right)$. Then $\forall \bar{\lambda}_k,\bar{\mu}_j$,
\begin{align}
\Tilde{g}\left(\{\bar{\lambda}_k\},\{\bar{\mu}_j\}_{j=1}^{K_t-1}\right)
=\underset{\mathbf{S}_k\succeq \mathbf{0},\forall k}{\max}
& \Bigg\{  \sum\limits_{k=1}^{K_r}\bar{\lambda}_k \big[\mathrm{log}|\mathbf{I}+\mathbf{H}_k\mathbf{S}_k\mathbf{H}_k^H|-\mathrm{Tr}(\mathbf{F}_k\sum_{i\neq k}^{K_r}\mathbf{S}_i)-R_k^{BD}\big]  \notag \\
&   +\sum_{j=1}^{K_t-1}\bar{\mu}_j \sum_{k=1}^{K_r}\mathrm{Tr}\big[(\mathbf{B}_{K_t}-\mathbf{B}_j)\mathbf{S}_k\big]-\frac{1}{P}\sum_{k=1}^{K_r}\mathrm{Tr}(\mathbf{B}_{K_t}\mathbf{S}_k) \Bigg\} \label{E:equality1} \\
\geq   \sum\limits_{k=1}^{K_r} & \bar{\lambda}_k \big[\mathrm{log}|\mathbf{I}+  \mathbf{H}_k\mathbf{S}_k^{\star}\mathbf{H}_k^H|-\mathrm{Tr}(\mathbf{F}_k\sum_{i\neq k}^{K_r}\mathbf{S}_i^{\star})-R_k^{BD}\big] \notag \\
&   +\sum_{j=1}^{K_t-1}\bar{\mu}_j \sum_{k=1}^{K_r}\mathrm{Tr}\big[(\mathbf{B}_{K_t}-\mathbf{B}_j)\mathbf{S}_k^{\star}\big]-\frac{1}{P}\sum_{k=1}^{K_r}\mathrm{Tr}(\mathbf{B}_{K_t}\mathbf{S}_k^{\star}) \label{E:inequality1}\\
=\quad \Tilde{g}& \left( \{\lambda_k\}, \{\mu_j\}_{j=1}^{K_t-1}\right)  + \sum_{k=1}^{K_r}s_{\lambda_k}(\bar{\lambda}_k-\lambda_k)
+\sum_{j=1}^{K_t-1}s_{\mu_j}(\bar{\mu}_j-\mu_j), \label{E:equality2}
\end{align}
where equality \eqref{E:equality1} follows from \eqref{ReducedLagrangian}, inequality \eqref{E:inequality1} follows since $\Tilde{g}\left(\{\bar{\lambda}_k\},\{\bar{\mu}_j\}_{j=1}^{K_t-1}\right)$ is the maximum value over all $\mathbf{S}_k\succeq\mathbf{0}$ for the given dual variables $\{\bar{\lambda}_k\}$ and $\{\bar{\mu}_j\}_{j=1}^{K_t-1}$.  $s_{\lambda_k}$ and $s_{\mu_j}$ are given by \eqref{E:subgradientLambda} and \eqref{E:subgradientMu}, respectively. Equality \eqref{E:equality2} is obtained by using $\bar{\lambda}_k=(\bar{\lambda}_k-\lambda_k)+\lambda_k, \bar{\mu}_j=(\bar{\mu}_j-\mu_j)+\mu_j$ and the fact that $\{\mathbf{S}_k^{\star}\}$ achieves the maximum value $\Tilde{g}\left( \{\lambda_k\}, \{\mu_j\}_{j=1}^{K_t-1}\right)$.  Then together with \eqref{E:subgDef}, Lemma~\ref{Lemma:subgradients} follows.

\bibliographystyle{IEEEtran}
\bibliography{IEEEabrv,IEEEfull}

\begin{figure}[htb]
\centering
\includegraphics[width=4in, height=3in]{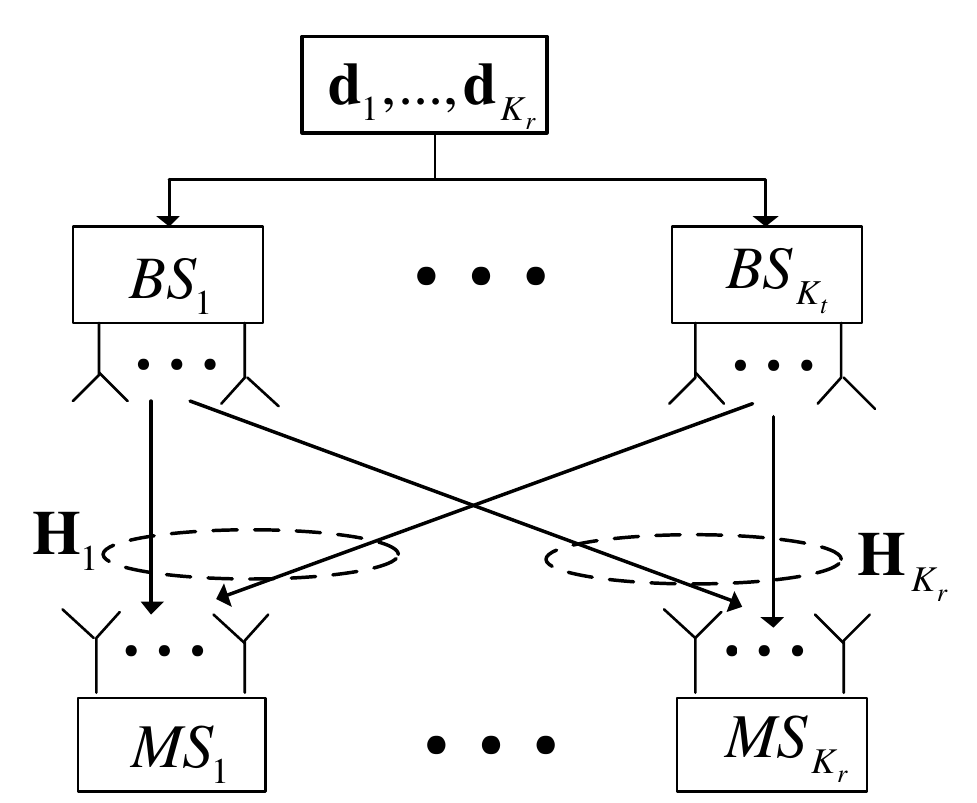}
\caption{System Model}
\label{F:fig1}
\end{figure}
%
%
%
\begin{figure}
\centering
\includegraphics[width=13cm, height=9.5cm]{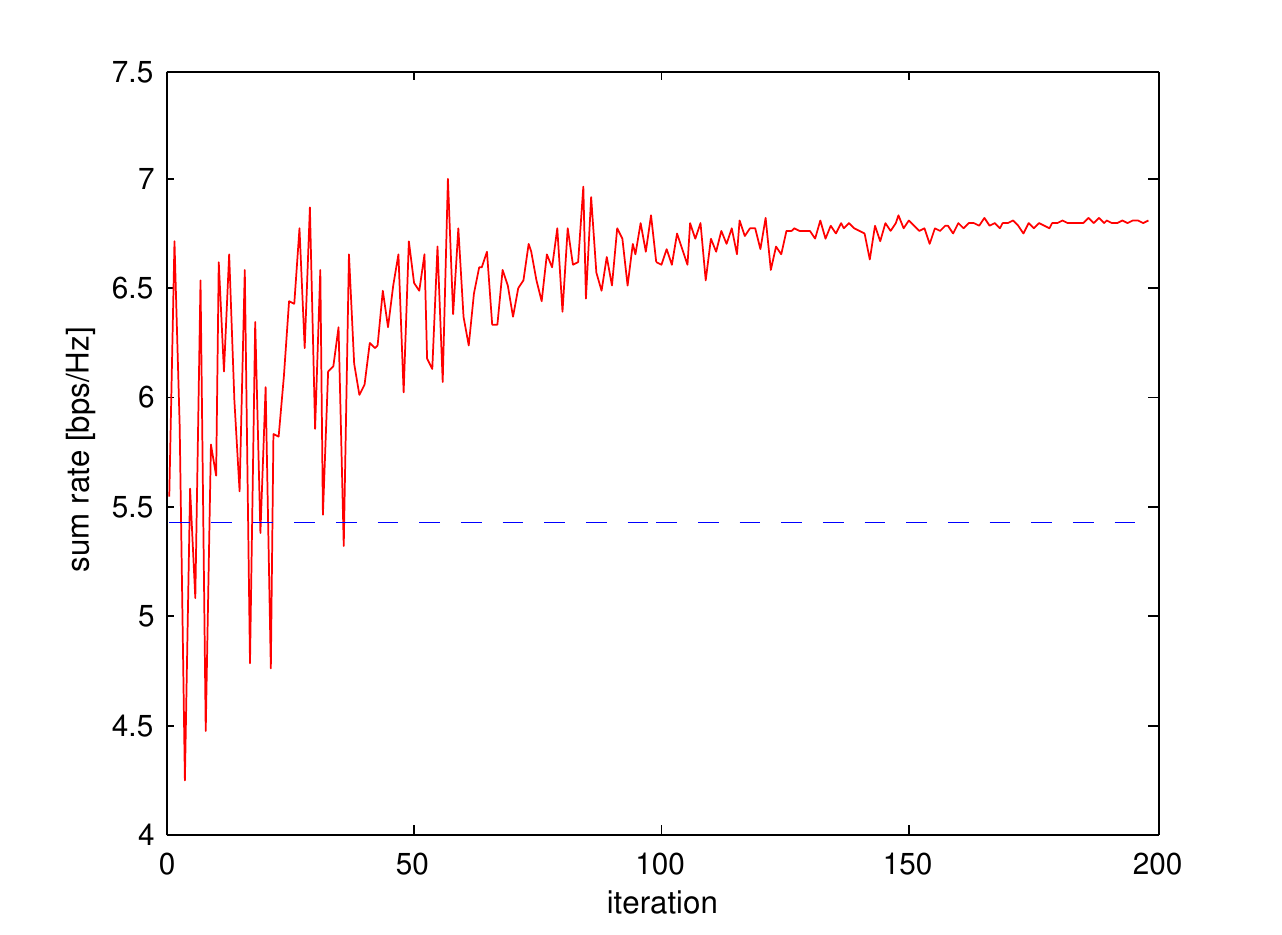}
\caption{Convergence behavior of Algorithm $1$. $[K_t \ N_t \ K_r \ N_r]=[3 \ 2 \ 3 \ 2], SNR = 10 $ dB. Dashed line shows the achieved sum rate with BD.}
\label{F:fig3}
\end{figure}
\begin{figure}
\centering
\includegraphics[width=13cm, height=9.5cm]{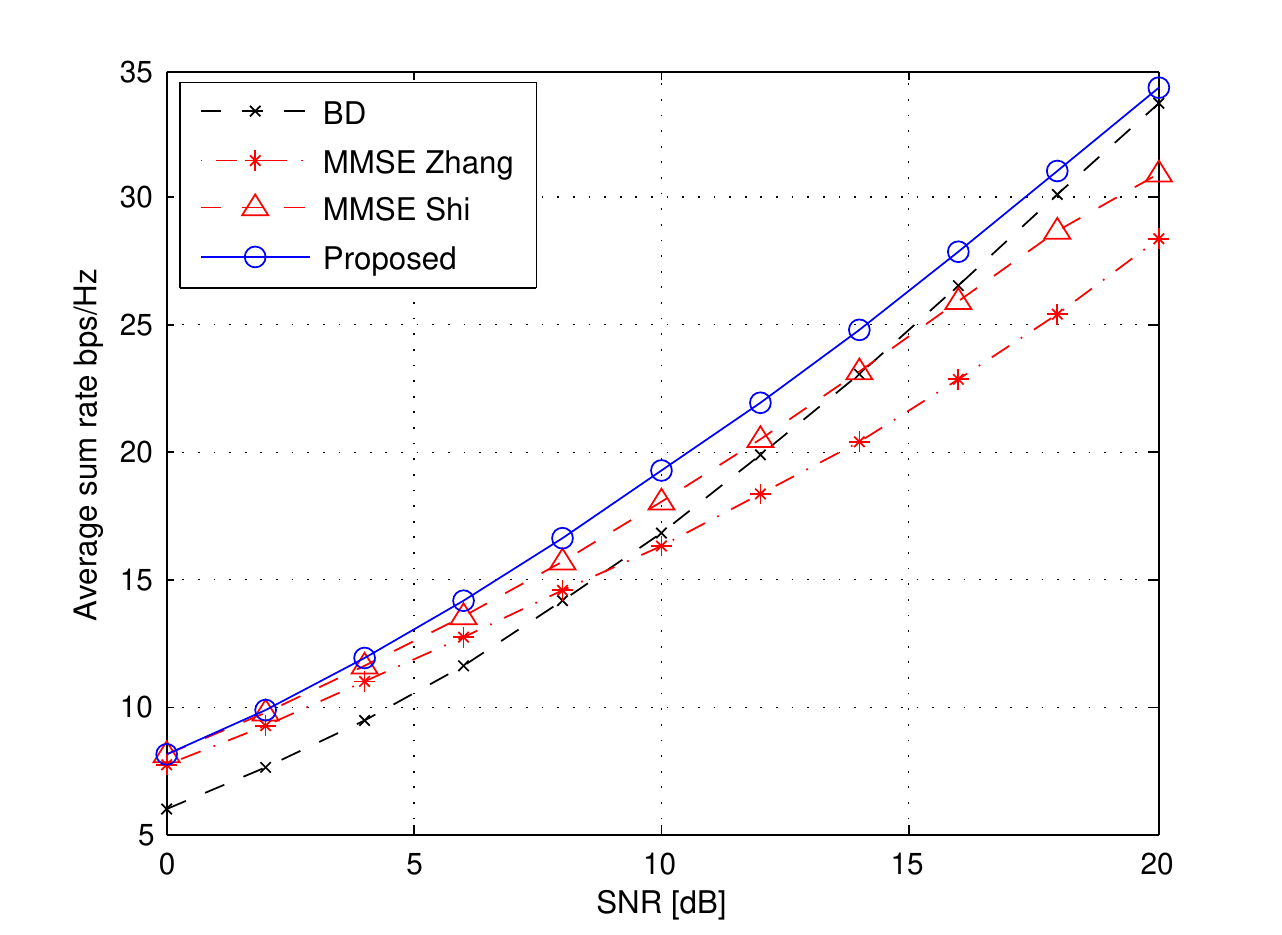}
\caption{Average sum rate for various schemes, $[K_t \ N_t \ K_r \ N_r]=[3 \ 2 \ 3 \ 2]$.}
\label{F:fig4}
\end{figure}
\begin{figure}
\centering
\includegraphics [width=13cm,height=9.5cm]{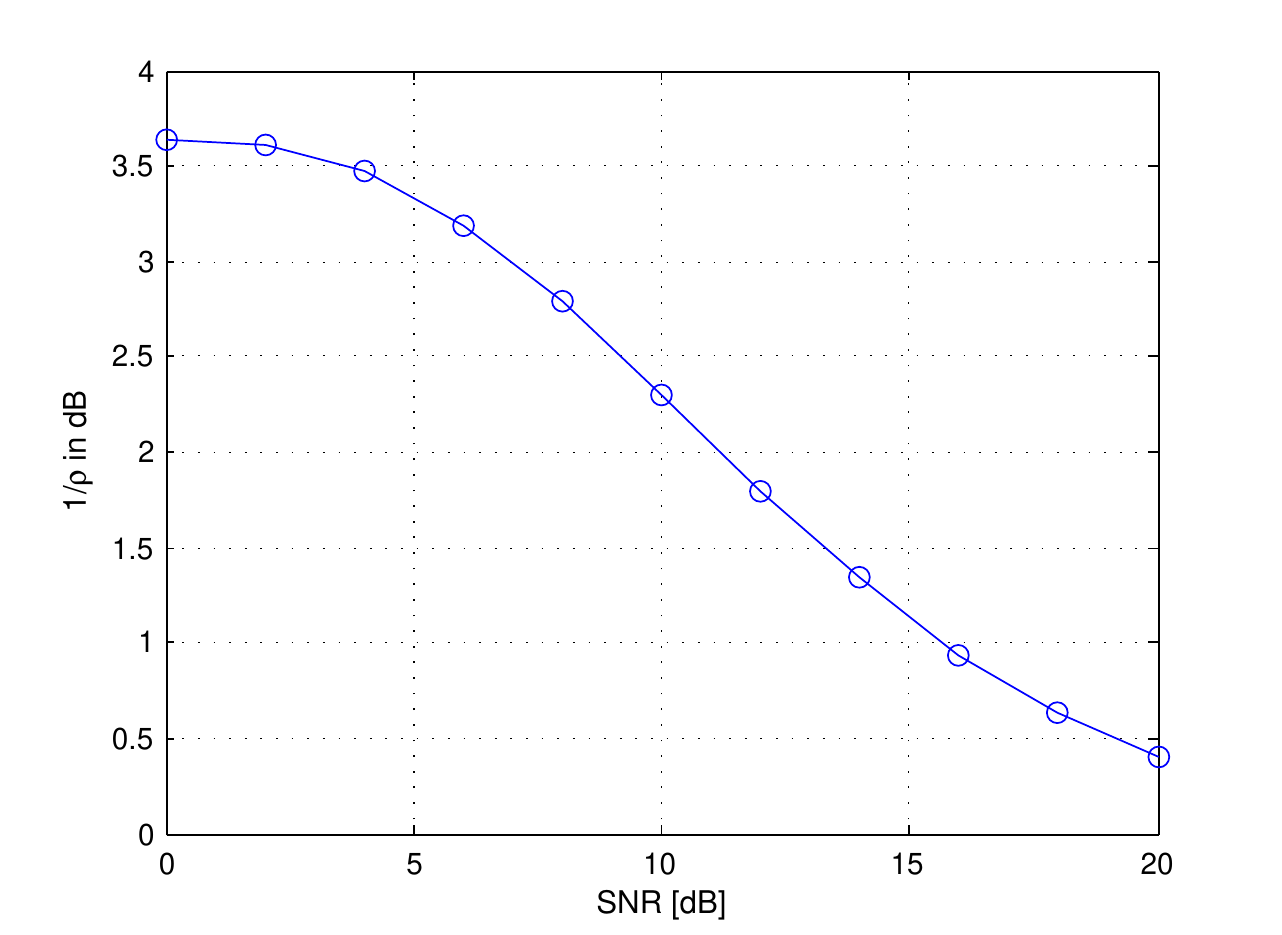}
\caption{Approximated effective SNR enhancement $1/\rho^{\star\star}$ in dB, $[K_t \ N_t \ K_r \ N_r]=[3 \ 2 \ 3 \ 2]$.}
\label{F:fig5}
\end{figure}
\end{document}